\documentclass[12pt,final,
thmsa,epsf,a4paper]
{article}
\usepackage{booktabs}

\newtheorem{theorem}{Theorem}

\newtheorem{corollary}{Corollary}
\newtheorem{definition}{Definition}
\newtheorem{example}{Example}
\newtheorem{lemma}{Lemma}
\newtheorem{proposition}{Proposition}
\newtheorem{remark}{Remark}

\newenvironment{proof}[1][Proof]{\emph{#1.} }{\  \hfill $\square $ \vspace{5 pt}}

\usepackage{stackrel}
\usepackage{natbib}
\usepackage{amssymb}
\usepackage[dvips]{graphics}
\usepackage{amsmath}

\usepackage{mathpazo}

\usepackage{enumerate}

\usepackage{amsfonts}

\usepackage{amssymb}

\usepackage{enumerate}

\usepackage{mathrsfs}

\usepackage[usenames]{color}

\usepackage{pgf,tikz}

\usetikzlibrary{decorations.pathreplacing,angles,quotes}
\usetikzlibrary{arrows.meta}
\usetikzlibrary{arrows, decorations.markings}

\usepackage{hyperref}
\hypersetup{
    colorlinks,
    citecolor=blue,
    filecolor=black,
    linkcolor=blue,
    urlcolor=blue
}

\usepackage{caption}
\usepackage{soul}
\usepackage{pgf,tikz}
\usetikzlibrary{arrows}

\usepackage[utf8]{inputenc}
\usepackage{amssymb, amsmath,geometry}
\usepackage{fancyhdr}

\usepackage{soul}
\usepackage{cancel}

\usepackage{emptypage}

\newcommand*\samethanks[1][\value{footnote}]{\footnotemark[#1]}

\begin{document}

\title{Counting steps for re-stabilization in a labor matching market\thanks{
We acknowledge the financial support of UNSL through grants 03-2016 and 03-1323, and from Consejo Nacional
de Investigaciones Cient\'{\i}ficas y T\'{e}cnicas (CONICET) through grant
PIP 112-200801-00655, and from Agencia Nacional de Promoción Cient\'ifica y Tecnológica through grant PICT 2017-2355.}}


\author{Agustín G. Bonifacio\thanks{Instituto de Matem\'{a}tica Aplicada San Luis (UNSL-CONICET) and Departamento de Matemática, Universidad Nacional de San
Luis, San Luis, Argentina. Emails: \texttt{abonifacio@unsl.edu.ar} (A. G. Bonifacio), \texttt{ncguinazu@unsl.edu.ar} (N. Gui\~{n}azu), \texttt{nmjuarez@unsl.edu.ar} (N. Juarez), \texttt{paneme@unsl.edu.ar} (P. Neme) and \texttt{joviedo@unsl.edu.ar} (J. Oviedo).} \and  Nadia Gui\~{n}az\'u \samethanks[2] \and Noelia Juarez\samethanks[2] \and Pablo Neme\samethanks[2] \and Jorge Oviedo\samethanks[2]}

\date{\today}
\maketitle

\begin{abstract}

We study a one-to-one labor matching market. If a worker considers resigning from her current job to obtain a better one, how long does it take for this worker to actually get it? We present an algorithm that models this situation as a re-stabilization process involving a vacancy chain. Each step of the algorithm is a link of such a chain. We show that the length of this vacancy chain, which can be interpreted as the time the worker has to wait for her new job, is intimately connected with the lattice structure of the set of stable matchings of the market. Namely, this length can be computed by considering the cardinalities of cycles in preferences derived from the initial and final stable matchings involved.


\bigskip

\bigskip


\end{abstract}

 \section{Introduction}
\label{intro}

Since the seminal work of \cite{blum1997vacancy}, a 
large  literature has emerged studying re-stabilization processes in two-sided matching markets. Most of this literature has focused on disequilibrium situations triggered by the retirement of senior-level workers or the entry of new firms.  This perturbation of equilibrium unrolls a vacancy adjustment dynamic that leads to a new stable situation. However, some other sources of disequilibrium, and their consequences for the market functioning, have remained unexplored.


In this paper, we present the following problem. Assume that the market clears in a stable matching other than the worker-optimal. This implies room for improvement from the workers' point of view. If a worker considers resigning from her current job to obtain a better one, how long does it take for this worker to actually get it?
Note that this worker's resignation causes the firm that she left to make an offer to another worker, which may generate another vacant position in a different firm and thus trigger a vacancy chain. This vacancy chain continues until a new stable matching is eventually reached in which the worker that first resigned obtains a new (and better) job.
The length of this vacancy chain can be interpreted as the time the worker has to wait for her new job. Since she will be unemployed during this adjustment process, knowing its length will allow her to decide whether to resign and wait for the new job or to remain in her \emph{status quo}.

 To solve this problem, we present an algorithm that models the aforementioned vacancy chain process. At each stage of this algorithm, a worker-firm pair is matched, modifying the previously obtained overall matching between workers and firms. We demonstrate that the algorithm terminates in a finite number of steps and yields a new stable matching as output. This solution is intimately connected with the lattice structure of the set of stable matching with respect to the unanimous ordering for workers \citep{knuth1976marriages}. Beginning from a stable matching (different from the worker-optimal) and through a procedure of ``reduction of preferences'', we can define a ‘‘cycle in preferences’’ that allows us to generate a new matching, called a ‘‘cyclic matching’’, that turns out to be stable \citep{irving1986complexity}.\footnote{To the best of our knowledge, \cite{irving1986complexity} is the first paper that introduces the notion of cycles in preferences, under the name of ``rotations''.} A cycle in preferences is a set of worker-firm pairs that carries the information on how to modify the original matching in order to obtain its derived cyclic matching. Moreover, given two stable matchings related under this unanimous ordering, one can find a sequence of stable matchings from one to the other. Each stable matching of the sequence (but the first one) can be obtained from the previous one as a cyclic matching. 
 From this sequence of cyclic matchings, we can extract a sequence of cycles in preferences.

 Our algorithm exploits this lattice structure in the following way. For a sequence of cycles generated from the input and output stable matchings of the algorithm, 
 and when all acceptable worker-firm pairs are stable, we show that: (i) at each stage of the algorithm, the pair that is matched belongs to a cycle of the sequence; (ii) all pairs from the last cycle of the sequence are matched; and (iii) for each remaining cycle, all pairs except one are matched at some stage of the algorithm. This implies that we can compute the number of steps of the algorithm (i.e., the length of the vacancy chain) in terms of the cardinality of the cycles involved: it equals the sum of the cardinality of the last cycle and the cardinality \emph{minus one} of all the remaining cycles.



Besides the aforementioned seminal paper of \cite{blum1997vacancy}, where in a one-to-one model a decentralized process using the ``deferred acceptance algorithm''\footnote{For a one-to-one model, \cite{gale1962college} introduce this algorithm and prove that it always terminates in a stable matching,  thus showing that the stable set is non-empty.} mimics a vacancy chain dynamic that achieves stability, there are other papers that study re-stabilization processes triggered by the retirement of senior-level workers or the creation of new firms. For example, a similar process with an adequate version of the deferred acceptance algorithm can be adapted to a many-to-one setting \citep{cantala2004restabilizing}. 
 Another approach to address this issue 
is based on Tarski's Fixed Point Theorem. In several works, the vacancy chain generated by the market destabilization is modeled by iterating Tarski operators having stable matchings as their fixed points \citep[see][among others]{bonifacio2022lattice,bonifacio2024lattice,kamada2023fair,wu2018lattice}. It is important to highlight that, unlike previous work, we do not assume changes in the size of the population in our market. Instead, we consider a scenario where a worker, without exiting the market, destabilizes it by resigning from her current job, thereby triggering a vacancy chain that results in a new stable matching where this worker is better off.

The notion of cycles in preferences has been widely exploited in the literature, mainly for computing the full set of stable matchings. The first works in this direction (for a one-to-one model) are those of \cite{irving1986complexity} and \cite{gusfield1987three}.\footnote{See also the treatment of cycles in preferences in the classic book by  \cite{roth1992two}.}  For many-to-one and many-to-many models with responsive preferences, the concepts were extended in \cite{cheng2008unified} and \cite{bansal2007polynomial}, respectively.  Finally, for a more general many-to-many model  with substitutable preferences that satisfy the law of aggregate demand, the concept was adapted by 
\cite{bonifacio2022cycles}. 
Furthermore, cycles in preferences have also been exploited in another direction. The set of probabilistic matchings fulfilling a notion of strong stability has been characterized using cycles in preferences and cyclic matchings for a one-to-one model \citep{neme2019characterization}  and for a many-to-one model with responsive preferences \citep{neme2021many}. 
    
The rest of the paper is organized as follows. The model and preliminary results are presented in Section \ref{seccion preliminares}. In Section \ref{seccion algoritmo}, an algorithm for re-stabilizing the market when a worker decides to resign in order to improve her labor situation is presented. The results on how many steps the algorithm takes to reach stability are gathered in Section \ref{seccion contando pasos}. Finally, some final remarks are presented in Section \ref{seccion final remarks}.
All proofs are relegated to the Appendix \ref{apendice pruebas}. 
In Appendix \ref{apendiceb},  we discuss the validity of our results without the assumption that all acceptable pairs are stable. Moreover, we show that the extension to a many-to-one model is not straightforward.

\section{Model and preliminaries}\label{seccion preliminares}
 
 We consider one-to-one matching markets where there are two disjoint sets of agents: the set of \textit{firms} $F$ and the set of \textit{workers} $W$. Each agent $a\in F\cup W$ has a strict preference relation $P_a$ over the agents on the
 other side of the market and the prospect of being unmatched, denoted by $\emptyset$.  For each agent $a\in F\cup W$, $R_a$ is the weak preference associated with $P_a$. Let $P_F$ be the preference profile for all firms, and let $P_W$ be the preference profile for all workers. We denote by $P$ the preference profile for all agents.  Since the sets $F$ and $W$ are kept fixed throughout the paper, we often identify the market $(F,W,P)$ with the preference profile $P$.  Given an agent $a \in F\cup W$, an agent
 $b$ in the opposite side of the market is \textbf{acceptable for  $\boldsymbol{a}$ under $\boldsymbol{P}$} if $bP_a\emptyset$. A pair $(w, f) \in W \times F$ is an \textbf{acceptable pair under $\boldsymbol{P}$ } if $f$ is acceptable for $w$ under $P$ and
 $w$ is acceptable for $f$ under $P$. In this paper, the preference relation $P_a$ is represented  by the ordered list of its acceptable agents (from most to least preferred).\footnote {For instance, $P_f:w_1,w_4,w_2$ indicate that $w_1P_fw_4P_fw_2P_f\emptyset$ and $P_w:f_1,f_2,f_3$ indicate that $f_1P_wf_2P_wf_3P_w\emptyset$.}

A \textbf{matching} $\mu$ is a function from the set $F\cup W$ into $F\cup W\cup 
\left\lbrace \emptyset\right\rbrace $ such that for each $w\in W$ and for each $f\in F$ (i) $\mu(w)\in F\cup 
\left\lbrace \emptyset\right\rbrace$, (ii) $\mu(f)\in W\cup 
\left\lbrace \emptyset\right\rbrace$, and (iii) $w= \mu(f)$ if and only if $f=\mu(w)$.

Agent $a\in F\cup W$ is \textbf{matched} if $\mu(a) \neq \emptyset$, otherwise she is \textbf{unmatched}. For the following definitions, fix a preference profile $P$.
A matching $\mu$ is \textbf{blocked by agent $\boldsymbol{a}$} if $\emptyset P_a \mu(a)$. A matching is  \textbf{individually rational} if it is not blocked by any individual agent. A matching $\mu$ is \textbf{blocked by a worker-firm pair $\boldsymbol{(w,f)}$} if, $f P_w \mu(w),$  and $w P_f \mu(f)$. A matching $\mu$ is \textbf{stable} if it is not blocked by any individual agent or any worker-firm pair. The set of stable matchings  for a preference profile $P$ is denoted by $\boldsymbol{S(P)}.$ 

Given a preference profile $P$ and two  matchings $\mu$ and  $\mu'$,  we write  $\boldsymbol{\mu P_F \mu'}$ whenever $\mu(f) R_f \mu'(f)$ for each $f\in F$, and there is $f'\in F$ such that $\mu(f') P_{f'}\mu'(f')$. We write  $\boldsymbol{\mu R_F \mu'}$ whenever $\mu(f) R_f \mu'(f)$ for each $f\in F$. Similarly, we write $\boldsymbol{\mu P_W\mu'}$ and $\boldsymbol{\mu R_W\mu'}$. Note that $R_F$ and $R_W$, which represent the common preferences of the firms and workers, respectively, are partial orders over the set of matchings.

\citet{knuth1976marriages} established that the stable matchings set has a dual lattice structure with respect to the partial orders $R_F$ and $R_W$. These lattices contain two distinctive matchings: the \textbf{firm-optimal stable matching} with respect to $R_F$, denoted by  $\boldsymbol{\mu_F}$, and the \textbf {worker-optimal stable matching}  with respect to $R_W$, denoted by $\boldsymbol{\mu_W}$. Moreover, matching $\mu_F$ is the \textbf {worker-pessimal stable matching}  with respect to $R_W$, and $\mu_W$ is the \textbf{firm-pessimal stable matching} with respect to $R_F$  \citep[see][for more details]{roth1992two}.

\subsection{The reduction procedure}

We present the preference reduction procedure initially presented by \cite{irving1986complexity}. 
This procedure will allow us to define a cycle in preferences, an essential concept for counting the stages of our algorithm.  
Given a market $P$, let $\mu ,\widetilde{\mu}\in S(P)$ such that $\mu P_F \widetilde{\mu}$. 
The reduction procedure is described as follows.

\begin{enumerate}
\item[\textbf{Step 1:}] 

\begin{enumerate}
[(a)]

\item Remove all $w$ who are more preferred than
$\mu(f)  $ from $f's$ list of acceptable workers. 
\item Remove all $f$ who are more preferred than $\widetilde{\mu}(
w)$ from $w's$ list of acceptable firms.
\end{enumerate}

\item[\textbf{Step 2:}] 

\begin{enumerate}
[(a)]
\item Remove all $f$ who are less preferred than
$\mu(w)$ from $w's$ list of acceptable firms.
\item Remove all $w$ who are less preferred than $\widetilde{\mu}(
f)$ from $f's$ list of acceptable workers.
\end{enumerate}

\item[\textbf{Step 3:}] After steps 1 and 2, if $f$ is not acceptable for $w$ (i.e., if
$f$ is not on $w's$ preference list as now modified), then
remove $w$ from $f's$ list of acceptable workers, and similarly, remove from  $w's$ list of acceptable firms any $f$ to whom $w$ is no longer acceptable.
\end{enumerate}

\noindent The  profile obtained by this procedure  is called the \textbf{reduced preference profile with respect to $\boldsymbol{\mu$ and $\widetilde{\mu}}$}, and is denoted by $\boldsymbol{P^{\mu,\widetilde{\mu}}}$. When $\widetilde{\mu}=\mu_W,$ the  profile is simply called the \textbf{reduced preference profile with respect to $\boldsymbol{\mu}$}, and is denoted by $\boldsymbol{P^{\mu}}$.
\begin{remark}
\label{remark de reduccion M-M} 
Let $P$ be a market and assume $\mu,\widetilde{\mu}\in S(P)$ with $\mu P_F \widetilde{\mu}$. Then, the following statements hold:
\begin{enumerate}[(i)]
\item for each $f \in F$, $\mu(f)$ and $\widetilde{\mu}(f)$ are the first and last entries 
  in $f's$ reduced preference, respectively. Symmetrically, for each $w \in W$, $\widetilde{\mu}(w)$ and $\mu(w)$ are
the first and last entries in $w's$ reduced preference, respectively. 

\item $\mu $ is the firm--optimal stable matching under $P^{\mu,\widetilde{\mu} }$ and $%
\widetilde{\mu}$ is the worker--optimal stable matching under $P^{\mu,\widetilde{\mu} }$.
Furthermore, $\widetilde{\mu}$ is the firm--pessimal stable matching under $P^{\mu,\widetilde{\mu} }$ and $\mu $ is the worker--pessimal stable matching under $P^{\mu,\widetilde{\mu} }
$.

\item $f$ is acceptable to $w$ under 
$P^{\mu,\widetilde{\mu} }$  if and only if $w$ is acceptable to $f$ under 
$P^{\mu,\widetilde{\mu} }.$ 

\end{enumerate}
\end{remark}

The following proposition is taken from \cite{roth1992two}.
\begin{proposition}\label{establePMUsiysolosireducido}
     Let $\mu,\widetilde{\mu}\in S(P)$ with $\mu P_{F}\widetilde{\mu}$. Then, $\mu'\in S(P)$ and  $\mu P_{F} \mu' P_{F}\widetilde{\mu}$ if and only if  $\mu' \in S(P^{\mu,\widetilde{\mu}}).$ 
\end{proposition}

The following example illustrates the reduction procedure for a matching market.

 \begin{example}\label{Ejemplo 1}
Consider a market where $F=\{f_1,f_2,f_3,f_4\}$, $W=\{w_1,w_2,w_3,w_4\}$, and  the preference profile is given 
by:
\begin{center}
\begin{tabular}{l}
$P_{f_1}:w_1,w_2,w_3,w_4$\\
$P_{f_2}:w_2,w_1,w_4,w_3$\\
$P_{f_3}:w_3,w_4,w_1,w_2$\\
$P_{f_4}:w_4,w_3,w_2,w_1$\\
\end{tabular}~~~~~~~~~~~~~~~~~~~
\noindent\begin{tabular}{l}
$P_{w_1}:f_4,f_3,f_2,f_1$\\
$P_{w_2}:f_3,f_4,f_1,f_2$\\
$P_{w_3}:f_2,f_1,f_4,f_3$\\
$P_{w_4}:f_1,f_2,f_3,f_4$\\
\end{tabular}
 \end{center}
\noindent Let us consider the following stable matchings

$$
\mu=\begin{pmatrix}

f_1 & f_2 & f_3 &f_4  \\
 w_2& w_1 &  w_3& w_4  \\
 
\end{pmatrix}
\text{ and }
\mu_W=\begin{pmatrix}

f_1 & f_2 & f_3 & f_4 \\
w_4& w_3 & w_2& w_1 \\
 
\end{pmatrix}.$$
\medskip

\noindent  After the reduction procedure is performed, the reduced preference profile with respect to $\mu$ and $\mu_W$ is presented in Table \ref{tabla reducida del ejemplo}.
\begin{table}[h!]
\centering
  \begin{tabular}{l}
$P^{\mu}_{f_1}:w_2,w_3,w_4$\\
$P^{\mu}_{f_2}:w_1,w_4,w_3$\\
$P^{\mu}_{f_3}:w_3,w_4,w_1,w_2$\\
$P^{\mu}_{f_4}:w_4,w_3,w_2,w_1$\\
\end{tabular}~~~~~~~~~~~~~~~~~~~
\begin{tabular}{l}
$P^{\mu}_{w_1}:f_4,f_3,f_2$\\
$P^{\mu}_{w_2}:f_3,f_4,f_1$\\
$P^{\mu}_{w_3}:f_2,f_1,f_4,f_3$\\
$P^{\mu}_{w_4}:f_1,f_2,f_3,f_4$\\
\end{tabular}
\caption{Reduced preference profile $P^\mu$.}
\label{tabla reducida del ejemplo}
\end{table}\hfill $\Diamond$
\end{example}

\section{Algorithm}\label{seccion algoritmo}
Consider a labor market where all agents are matched under a stable matching and that such stable matching is not the worker-optimal. Thus, there is a worker and another stable matching in which she could be better off. This worker can resign her position in order to wait for a better offer. We show in this section that such a better offer always arrives and that the market is re-stabilized into a matching in which no worker is worse off. In each stage of the algorithm, the firm left alone in the previous stage chooses the best worker willing to be employed with this firm (such a worker always exists). This generates a new empty position in another firm. This process continues until the worker who disrupted the initial stability of the market receives an offer. Of course, this offer improves upon her initial situation.

Given a preference profile $P$ and a stable matching $\mu\in S(P)\setminus \{\mu_W\}$, the algorithm takes as inputs a worker $w^0$ such that $\mu(w^0)\neq\mu_W(w^0)$ and the reduced preference profile  $P^\mu$. We formally present the algorithm in Table \ref{tabla algoritmo}. 
\begin{table}[h!]
\centering
\begin{tabular}{l l}
\hline \hline
\multicolumn{2}{l}{\textbf{Algorithm:}}\vspace*{10 pt}\\

\textbf{Input} & A reduced profile $P^\mu$ and $w^0\in W$ such that $\mu(w^0)\neq\mu_W(w^0)$.\\

\textbf{Output} & A set of worker-firm pairs $A$, and a matching $\nu\in S(P^\mu)$.\vspace*{10 pt}\\
 & \hspace{-55pt}Define:\\
 & \hspace{-55pt} $f^0=\mu(w^0)$,\\
  & \hspace{-55pt} $ \nu^0  (f) =\left\{ 
\begin{array}{lcl}
 \emptyset  &  & \text{if } f =f^0 \\ 
\mu (f)&  & \text{otherwise}  \\ 
\end{array}%
\right. ,
$\\
& \hspace{-55pt} $\nu^0(w)=\left\{ 
\begin{array}{lcl}
f  &  & \text{if } w= \nu^0(f) \\ 
\emptyset &  & \text{otherwise}\\ 
\end{array}%
\right.,$ \\
& \hspace{-55pt} $A^0=\emptyset$.\\
\textbf{Stage $\boldsymbol{t\geq1}$} &  Let $W^t=\{w\in W\setminus\{w^0\}:f^{t-1}P^{\mu}_w \nu^{t-1}(w)\}$\\
& \hspace{20pt}\texttt{IF} $f^{t-1}P^{\mu}_{w^0} \mu(w^0)$\\
& \hspace{40 pt}\texttt{THEN} $W^t=W^t\cup\{w^0\}$\\
& Choose $w^t\in W^t$ such that $ w^t R^{\mu}_{f^{t-1}}w$ for each $w\in W^t$\\
& \hspace{20 pt}\texttt{IF} $w^t\neq w^0:$  \\
& \hspace{40 pt}\texttt{THEN} Define:\\
& \hspace{68 pt}$f^{t}=\nu^{t-1}(w^t)$  \\
&\hspace{68 pt}$
 \nu^t  (f) =\left\{ 
\begin{array}{lcl}
 \emptyset &  & \text{if } f =f^{t} \\ 
 w^{t}   &  & \text{if } f =f^{t-1} \\ 
\nu^{t-1} (f)&  & \text{otherwise}  \\ 
\end{array}%
\right. 
$\\
&\hspace{68 pt}$\nu^{t}(w)=\left\{ 
\begin{array}{lcl}
f  &  & \text{if } w= \nu^t(f) \\ 
\emptyset &  & \text{otherwise}\\ 
\end{array}%
\right.$ \\
&\hspace{68 pt}$A^t=A^{t-1}\cup \{(w^t,f^{t-1})\}$\\
&\hspace{40 pt}\texttt{AND} continue to Stage $t+1.$\\
& \hspace{20 pt}\texttt{ELSE}: Set $
 \nu  (f) =\left\{ 
\begin{array}{lcl}
  w^{0}   &  & \text{if } f =f^{t-1} \\ 
\nu^{t-1} (f)&  & \text{otherwise}  \\ 
\end{array}%
\right. 
$ \\
&\hspace{68 pt}$\nu  (w) =\left\{ 
\begin{array}{lcl}
f  &  & \text{if } w= \nu(f) \\ 
\emptyset &  & \text{otherwise}\\ 
\end{array}%
\right. 
$\\
&\hspace{68 pt}$A=A^{t-1}\cup\{(w^0,f^{t-1})\}$, and \texttt{STOP}. \\


\hline \hline
\end{tabular}

\caption{Re-stabilization algorithm}
\label{tabla algoritmo}
\end{table}

\bigskip

\begin{remark}\label{remark algoritmo}
Concerning the algorithm in Table \ref{tabla algoritmo}, notice that:
\begin{enumerate}[(i)]
    
    \item  If $\nu^t(w^0)=\emptyset$, then $(w^0,f^0)$ blocks $\nu^{t}$ under $  P^{\mu}$. Thus, $\nu^t$ is not a stable matching under $  P^{\mu}$. However, each $\nu^t$  is individually rational under $  P^{\mu}$.

    \item  For each stage $t$ of the algorithm, $\nu(w)R^{\mu}_w\nu^t(w)R^{\mu}_w\nu^{t-1}(w)$ for each $w\in W$. 
   
    \item For each $w\in W$ such that $\mu(w)\neq \nu(w)$, there is a stage $t$ of the algorithm in which $w$ is chosen, i.e., there is a stage $t$ such that $w=w^t.$
\end{enumerate}
  
\end{remark}

Next, we demonstrate that the algorithm is well-defined. The following theorem shows that: (i) for each stage of the algorithm, there is always at least one worker willing to accept the position in the offering firm (that was abandoned in the previous stage); (ii)  the worker who disrupted the initial stability of the market receives an improving offer, generating that the algorithm stops; (iii) the market always reaches stability once the algorithm stops.

\begin{theorem}\label{nu es estable}
For the algorithm presented in Table \ref{tabla algoritmo}, we have that:
\begin{enumerate}[(i)]
    \item For each stage $t$ of the algorithm, $W^t\neq \emptyset.$
 \item The algorithm stops in a finite number of stages.
 \item Let $\nu$ be the matching obtained by the algorithm, then $\nu\in S(P^\mu).$
\end{enumerate}

\end{theorem}



As, by Theorem \ref{nu es estable} (iii), $\nu \in S(P^{\mu})$, we have $\mu P^{\mu}_F \nu R^{\mu}_F \mu_W$. Then, the reduction procedure implies $\mu P_F \nu R_F \mu_W$. Therefore, applying Proposition \ref{establePMUsiysolosireducido} the next result follows.

\begin{corollary}\label{corolario nu estable}
    Let $\mu$ and $\nu$ be the input and output of the algorithm, respectively. Then, $\nu\in S(P)$ and $\mu P_F \nu.$
\end{corollary}
The algorithm is illustrated for the market in Example \ref{Ejemplo 1} as follows.
  
  \bigskip
  
 \noindent \textbf{Example 1 (Continued)} 
\emph{We apply the algorithm to the reduced profile $P^{\mu}$ presented in Table \ref{tabla reducida del ejemplo}, and considering $w^0=w_1.$ Define $f^0=\mu(w_1)=f_2$, $A^0=\emptyset,$ and 
$$
\nu^0=\begin{pmatrix}
f_1 & f_2 & f_3 & f_4& \emptyset  \\
w_2& \emptyset & w_3& w_4 & w_1\\
 \end{pmatrix}.$$
\medskip
\noindent In what follows, we detail each of its stages:
\newline
\noindent \textit{\textbf{Stage 1}}
As $W^1=\left\lbrace w_3,w_4\right\rbrace $ and $w_4P^{\mu}_{f_2}w_3$, then $ w^1=w_4$, and $ f^1=\nu^0(w_4)=f_4.$ Since $ w^1=w_4\neq w^0=w_1,$ we have 
$A^1=\left\lbrace (w_4,f_2)\right\rbrace $, and
$$
\nu^1=\begin{pmatrix}
f_1 & f_2 & f_3 & f_4& \emptyset  \\
w_2& w_4 & w_3& \emptyset & w_1\\
\end{pmatrix}.$$ 
\bigskip
\noindent \textit{\textbf{Stage 2}}
As $W^2=\left\lbrace w_1,w_2,w_3\right\rbrace $ and $w_3P^{\mu}_{f_4}w_2P^{\mu}_{f_4}w_1$, then $ w^2=w_3$, and $ f^2=\nu^1(w_3)=f_3.$ Since $ w^2=w_3\neq w^0=w_1,$ we have $A^1=\left\lbrace (w_3,f_4),(w_4,f_2)\right\rbrace $, and 
$$
\nu^2=\begin{pmatrix}
f_1 & f_2 & f_3 & f_4& \emptyset  \\
w_2& w_4 & \emptyset & w_3 & w_1\\
\end{pmatrix}.$$ 
\bigskip
\noindent \textit{\textbf{Stage 3}}
Lastly, as $W^3=\left\lbrace w_1,w_2\right\rbrace $ and $w_1P^{\mu}_{f_3}w_2$, then $ w^3=w_1$, and $ w^3=w_1=w^0.$ Therefore, the algorithm stops, and the outputs of the algorithm are $A=\left\lbrace (w_1,f_3),(w_3,f_4),(w_4,f_2)\right\rbrace $, and  
$$
\nu=\begin{pmatrix}
f_1 & f_2 & f_3 & f_4  \\
w_2& w_4 & w_1 & w_3 \\
\end{pmatrix}.$$}\hfill $\Diamond$

\section{Counting steps for re-stabilization}\label{seccion contando pasos}
This section contains two subsections. Subsection \ref{subsection ciclos} presents the notions of cycle in preferences and cyclic matching that are essential for counting the stages that the algorithm takes to re-stabilize. The main result of the paper establishes the number of stages that the algorithm needs to re-stabilize the market and is presented in Subsection \ref{subseccion contando}. This result can be interpreted as how many stages a worker has to wait for a better position when she resigns.
\subsection{Cycles in preferences}\label{subsection ciclos}
In this subsection we present the notions of cycle in preferences and cyclic matching,  first presented in \cite{irving1986complexity} for computing the full set of stable matchings. 
Consider a stable matching $\mu \in S(P),$  and the reduced preference profile with respect to $\mu$, denoted by  $P^{\mu}$. Recall some important facts about this reduced profile: (i) $\mu(f)$ is $f$'s most preferred partner and $\mu_W(f)$  is $f$'s least preferred partner according to $P^{\mu}_f$, for each $f \in F$; and (ii)    $\mu_W(w)$ is $w$'s most preferred partner and $\mu(w)$  is $w$'s least preferred partner according to $P^{\mu}_w$, for each $w \in W.$    A  cycle for $P^{\mu}$ can be seen as an ordered sequence of worker-firm pairs $\{(w_{1},f_{0}),(w_{2},f_{1}),\ldots,(w_r,f_{r-1}),(w_0,f_{r})\}$ such that $w_i=\mu(f_i)$  and  $w_{i+1}$ is the second most-preferred worker for $f_i$ according to $P^{\mu}_{f_i}.$ Formally,

\begin{definition}\label{defino cyclo}
Let $\mu\in S(P)$. A \textbf{cycle $\boldsymbol{\sigma$ for $P^{\mu}}$} is an ordered  sequence of worker-firm pairs $\sigma=\{(w_{1},f_{0}),(w_{2},f_{1}),\ldots,(w_{r} ,f_{r-1}),(w_{0} ,f_{r})\}\subseteq W\times F$  such that, for $i=0,\ldots,r,$ we have:
\begin{enumerate}[(i)]
\item $w_i=\mu(f_i)\neq \mu_W(f_i),$
 \item $w_{i+1}R^{\mu}_{f_i}w$ for each $w\in W\setminus\{w_i\},$ and where $w_{r+1}=w_0.$ 
 \end{enumerate} 
\end{definition}


Let $\sigma $ be a cycle for $P^{\mu}.$ Denote by $\sigma_F$ and $\sigma_W$ to the set of  firms and workers involve in cycle $\sigma$ respectively. Given a cycle $\{(w_{1},f_{0}),(w_{2},f_{1}),\ldots,(w_{r} ,f_{r-1}),(w_{0} ,f_{r})\}$ for $P^\mu$ can be used to obtain a new matching from matching $\mu$ by breaking the partnership between firm $f_i$ and worker $w_i$ and establishing a new partnership between firm $f_i$ and worker $w_{i+1}$ for each $i=0,\ldots,r$ (modulo $r+1$), keeping all remaining partnerships in $\mu$ unaffected. This new matching is called a cyclic matching. Formally,

\begin{definition}
Let $P$ be a market. Given $\mu\in S(P)$, and the reduced preference profile $P^{\mu}$, let $\sigma=\{(w_{1},f_{0}),(w_{2},f_{1}),\ldots,(w_{r} ,f_{r-1}),(w_{0} ,f_{r})\}$ be a cycle for $P^{\mu}$. The \textbf{cyclic matching} of $\mu$ is defined as follows:

\medskip
$\mu_\sigma(f)=\left\{
\begin{array}{ll}
\mu(f_{i+1})&\text{if } f=f_i \text{ for }i=0,\ldots,r-1,\\
\mu(f_{0})& \text{if } f=f_r,\\
\mu(f) &\text{ for each }f\not\in{\sigma_F}.
\end{array}\right.$

\end{definition}

The resulting cyclic matching is stable not only in the reduced preference profile $P^\mu$, but also in the original preference profile $P$ \cite[see][for more details]{irving1986complexity,roth1992two}.

The following theorem is taken from \cite{gusfield1987three}.
\begin{theorem}\label{asignaciones consec}
Let $P$ be a market, and let $\mu_F$ and $\mu_W$ be the firm-optimal and worker-optimal stable matchings, respectively. Then, there is a sequence of stable matchings $\{\mu^0,\ldots,\mu^q\}$ such that $\mu^0=\mu_F$, $\mu^{i}$ is a cyclic matching of $\mu^{i-1}$ for $i=1,\ldots,q$, and $\mu^q=\mu_W$. Moreover, such a sequence may not be unique but all sequences contain the same cycles (probably in different order and only once). 
\end{theorem}

Let $\mu$ and $\nu$ be the input and output of the algorithm, respectively. Notice that, by the reduction procedure, $\mu$ and $\nu$ are the firm-optimal and worker-optimal stable matchings under $P^{\mu,\nu}$, respectively. A \textbf{sequence of cycles generated from $\boldsymbol{\mu$ to $\nu}$} is a sequence of cycles $\{\sigma^1,\ldots,\sigma^k\}$ such that there is a sequence of stable matchings $\{\mu^0,\ldots,\mu^k\}$ with $\mu^0=\mu$,  $\mu^{i}=\mu^{i-1}_{\sigma^{i}}$ where $\sigma^{i}$ is a cycle for $P^{\mu^{i-1},\nu}$ with for $i=1,\ldots,k$ and $\mu^k=\nu$. Theorem \ref{asignaciones consec} implies that such sequence exists.
Notice that the sequence of cycles generated from $\mu$ to $\nu$ may not be unique.

\begin{lemma}\label{el ultimo ciclo de dos suceciones es el mismo}
If two different sequences of cycles are generated from $\mu$ to $\nu$, the last cycle of each sequence is the
same.
 
\end{lemma}

Below, we return to Example \ref{Ejemplo 1} to illustrate the calculation of cycles and cyclic matchings in the reduced preference profile.

\medskip

\noindent \textbf{Example 1 (continued)}  Considering the reduced preference profile $P^\mu$, and the output of the algorithm $\nu\in S(P^\mu)$, we compute the reduced preference profile $P^{\mu,\nu}$ presented in Table \ref{tabla mu y nu reducida del ejemplo}.
\begin{table}[h!]
\centering
  \begin{tabular}{l}
$P^{\mu,\nu}_{f_1}:w_2$\\
$P^{\mu,\nu}_{f_2}:w_1,w_4$\\
$P^{\mu,\nu}_{f_3}:w_3,w_4,w_1$\\
$P^{\mu,\nu}_{f_4}:w_4,w_3$\\
\end{tabular}~~~~~~~~~~~~~~~~~~~
\begin{tabular}{l}
$P^{\mu,\nu}_{w_1}:f_3,f_2$\\
$P^{\mu,\nu}_{w_2}:f_1$\\
$P^{\mu,\nu}_{w_3}:f_4,f_3$\\
$P^{\mu,\nu}_{w_4}:f_2,f_3,f_4$\\
\end{tabular}
\caption{Reduced preference profile $P^{\mu,\nu}$.}
\label{tabla mu y nu reducida del ejemplo}
\end{table}\newline
\noindent We can compute cycle $\sigma^1=\{(w_4,f_3),(w_3,f_4)\}$ and the cyclic matching of $\mu$
$$\mu^1=\begin{pmatrix}
f_1 & f_2 & f_3 & f_4  \\
w_2& w_1 & w_4&  w_3\\
\end{pmatrix}.$$
\noindent To compute the next cycle in the sequence of cycles generated from $\mu$ to $\nu$, first we compute the reduced preference profile $P^{\mu^1,\nu}$ presented in Table \ref{tabla mu1 y nu reducida del ejemplo}.
\begin{table}[h!]
\centering
  \begin{tabular}{l}
$P^{\mu^1,\nu}_{f_1}:w_2$\\
$P^{\mu^1,\nu}_{f_2}:w_1,w_4$\\
$P^{\mu^1,\nu}_{f_3}:w_4,w_1$\\
$P^{\mu^1,\nu}_{f_4}:w_3$\\
\end{tabular}~~~~~~~~~~~~~~~~~~~
\begin{tabular}{l}
$P^{\mu^1,\nu}_{w_1}:f_3,f_2$\\
$P^{\mu^1,\nu}_{w_2}:f_1$\\
$P^{\mu^1,\nu}_{w_3}:f_4$\\
$P^{\mu^1,\nu}_{w_4}:f_2,f_3$\\
\end{tabular}
\caption{Reduced preference profile $P^{\mu^1,\nu}$.}
\label{tabla mu1 y nu reducida del ejemplo}
\end{table}\newline
\noindent Then, the next cycle is $\sigma^2=\{(w_4,f_2),(w_1,f_3)\}.$ Notice that the cyclic matching of $\mu^1$ is
$$\mu^2=\begin{pmatrix}
f_1 & f_2 & f_3 & f_4  \\
w_2& w_4 & w_1&  w_3\\
\end{pmatrix}=\nu.$$
\hfill $\Diamond$

\subsection{Counting steps}\label{subseccion contando}

In this subsection, we establish how many stages the algorithm takes in order to re-stabilize the market when worker $w^0$ decides to improve her labor situation. To do this, we make the crucial assumption that all acceptable pairs are stable. Given a market $P$, recall that a pair $(w,f)$ is \textbf{stable for $\boldsymbol{P}$}  if there is a stable matching $\mu$ such that $f=\mu(w)$ \citep[see, for instance,][]{gusfield1987three}. 
 First, we show that each pair that the algorithm computes belongs to one cycle of any sequence of cycles generated from the input matching to the output matching of the algorithm. Moreover, we demonstrate that the algorithm computes all 
 but one pair of each cycle in any sequence, except for the last cycle for which it computes the whole set of pairs.\footnote{By Lemma \ref{el ultimo ciclo de dos suceciones es el mismo}, the \emph{last} cycle is well-defined because it is the same for all sequences.}

Given a reduced preference profile $P^\mu$ and a worker $w^0$, denote by $A(P^\mu,w^0)$ to the output of the algorithm applied to the preference profile $P^{\mu}$ when $w^0$ decides to improve her labor situation.
The following lemma states that every worker-firm pair in $A(P^{\mu},w^0)$ belongs to exactly one cycle of the sequence of cycles generated from $\mu $ to $\nu $.

\begin{lemma}\label{cada par del algoritmo en un ciclo}
Let $\{\sigma^1,\ldots,\sigma^k\}$ be any sequence of cycles generated from $\mu$ to $\nu$. For each $(w,f)\in A(P^\mu, w^0)$ there is  exactly one  cycle $\sigma \in \{\sigma^1,\ldots,\sigma^k\}$ such that $(w,f)\in \sigma$. 
\end{lemma}

The following lemma states that: (i) each pair of the last cycle is computed by the algorithm ($A(P^\mu,w^0)$), and (ii) if the sequence of cycles generated from $\mu$ to $\nu$ has at least two cycles, all but one of the pairs in the first cycle of the sequence belong to $A(P^\mu,w^0)$. Formally, 



\begin{lemma}\label{cardinalidad del primer ciclo y ultimo ciclo}
    Let $\{\sigma^1,\ldots,\sigma^k\}$ be any sequence of cycles generated from $\mu$ to $\nu$. The following hold:
    \begin{enumerate}[(i)]
        \item  $\sigma^k\subseteq A(P^\mu,w^0).$
        \item If  $k\neq1,$ then $|A(P^\mu,w^0)\cap \sigma^1|=|\sigma^1|-1.$
    \end{enumerate}
\end{lemma}

Note that, item (i) is equivalent to $|A(P^\mu,w^0)\cap \sigma^k|=|\sigma^k|.$
In the following result, we determine exactly how many pairs are contained in the output set of the algorithm.

\begin{theorem}\label{teorema principal}
Let $ \mu\in S(P)$, $w^0\in W$ such that $\mu(w^0)\neq \mu_{W}(w^0)$, and let $\nu$ be the output of the algorithm applied to the profile $P^{\mu}$ and worker $w^0$. Then, 

    \begin{equation}\label{nombre}
        |A(P^\mu,w^0)|=|\sigma^k|+\sum_{i=1}^{k-1}(|\sigma^i|-1)
    \end{equation}

where $\{\sigma^1,\ldots,\sigma^k\}$ is any sequence of cycles generated from $\mu$ to $\nu$.
\end{theorem}

Thus, the number of steps for a vacancy chain triggered by the resignation of worker $w^0$ needed to reach stability is computed as in Theorem \ref{teorema principal}. The following example illustrates this result.\bigskip

\noindent \textbf{Example 1 (continued)} Recall that $A(P^\mu,w_1)=A=\left\lbrace (w_1,f_3),(w_3,f_4),(w_4,f_2)\right\rbrace $, $\sigma^1=\{(w_4,f_3),(w_3,f_4)\}$, and $\sigma^2=\{(w_4,f_2),(w_1,f_3)\}.$ Then, $|A(P^\mu,w_1)\cap \sigma^1|=1$, and $|A(P^\mu,w_1)\cap \sigma^2|=2.$ Therefore, $|A(P^\mu,w_1)|=3$.\hfill $\Diamond$ 

\section{Final remarks}\label{seccion final remarks}

Given a job matching market in equilibrium, if 
a worker resigns to obtain a better job, we present an algorithm that restores stability to the market. Notice that the  resigning worker triggers a unique vacancy chain, leading to a unique stable matching.  We show that the output of the algorithm not only recovers stability but also weakly improves the situation for all workers, leaving some of them strictly better off together with the one responsible for disrupting the original stability. By using cycles in preferences, we can count the steps necessary to restore such stability.  We finish with some final comments.


A crucial assumption in our main result is that all acceptable pairs are stable pairs. If some acceptable pairs are not stable, then Theorem \ref{teorema principal} is no longer valid. The reason is twofold. On one hand, some stage of the algorithm could select a worker-firm pair that is acceptable yet not stable, and thus such pair would not be part of any cycle in preferences. On the other hand, it could be the case that all the pairs in a cycle of the sequence different than the last one are included in the output of the algorithm, violating Lemma \ref{cardinalidad del primer ciclo y ultimo ciclo} (ii). However, the calculation involving the cardinalities of the cycles obtained in Theorem \ref{teorema principal}  still can be used as a lower bound for the amount of steps necessary to re-stabilize the market. We discuss these facts and present this weaker version of our result in  Appendix \ref{apendiceb}.



Concerning the cost of our algorithm, its operational time is linked to the operational time of computing all cycles in market $P^{\mu,\nu}$. \cite{gusfield1987three} proves, in his Theorem 5, that all cycles can be found in  $O(n^2)$ time, where $n$ is the number of agents on each side of the market. So, this result is also true for our algorithm  as a consequence of our Theorem \ref{teorema principal}. 

There are many papers studying re-stabilization processes in several matching markets. Some of these papers model a destabilization of the market by assuming that new workers arrive at the market, senior workers retire, firms create new positions, and firms downsize, among others. Each of these phenomena triggers either a layoff chain or a vacancy chain that eventually reaches stability \citep[see][among others]{bonifacio2022lattice,bonifacio2024lattice,blum1997vacancy,cantala2004restabilizing,wu2018lattice}. 
Each of these papers shares a common feature: the re-stabilization process occurs within the realm of quasi-stable matchings.\footnote{In short, a quasi-stable matching allows only blocking pairs (if any) where one of the agents of the blocking pair has an empty position. When firms have empty positions, the matching is referred to as \emph{firm-quasi-stable} or \emph{envy-free}. Otherwise, it is termed \emph{worker-quasi-stable}.} A natural question arises of whether the matchings $\nu^t$ are firm-quasi-stable or worker-quasi-stable. The answer is neither of them. Leaving aside $\nu$ that is stable, Lemma  \ref{par bloquenate de nu^t tambien de nu^t-1} (in Appendix \ref{apendice pruebas}) shows that  $(w^0,f^0)$ is a blocking pair for $\nu^t$ at each stage $t$ and $\nu^t(f^0)\neq\emptyset$, implying that $\nu^t$ is not a firm-quasi-stable matching. Moreover, Theorem \ref{nu es estable} (i) assures $W^t\neq \emptyset$ for each stage $t$. This implies that the pair $(w^t,f^{t-1})$ is a blocking pair for $\nu^{t-1}$ and $\nu^{t-1}(w^t)\neq \emptyset$. Then, $\nu^{t-1}$ is not a worker-quasi-stable matching.

\bibliographystyle{ecta}
\bibliography{biblio}

\appendix
\section{Proofs}\label{apendice pruebas}

The following lemma will be used in the proof of Theorem \ref{nu es estable}.
\begin{lemma}\label{par bloquenate de nu^t tambien de nu^t-1}
If $(w,f)$ is a blocking pair for $\nu^{t}$ under $P^{\mu}$, then either $(w,f)=(w^0,f^0)$ or $f=f^{t}$. 
\end{lemma}
\begin{proof}
    Let $(w,f)$ a blocking pair for  $\nu^{t}$ under $P^{\mu}$. Assume that  $(w,f)\neq (w^0,f^0)$, and $f\neq f^{t}$. 
  There are two cases to consider:

  \noindent \textbf{Case 1: $\boldsymbol{f=f^{t-1}.}$} The fact that $(w,f^{t-1})$ is a blocking pair for  $\nu^{t}$ and definition of $\nu^{t}$ imply that  
    \begin{equation}\label{ecu -2 Lemma par bloqueante de vt-1}
        wP^{\mu}_{f^{t-1}}\nu^{t}(f^{t-1})=w^{t}  
    \end{equation} 
    and 
    \begin{equation}\label{ecu -2 Lemma par bloqueante de vt-2}
        f^{t-1}P^{\mu}_{w}\nu^{t}(w).  
    \end{equation} 
    By \eqref{ecu -2 Lemma par bloqueante de vt-1} and definition of $W^t$ we have   $w\notin W^t$. Assume that $w=w^0$. Since $(w^0,f^{t-1})$ blocks $\nu^{t}$, $ w^0P^{\mu}_{f^{t-1}}\nu^{t}(f^{t-1})=w^{t}$. The fact that $t$ is not the last stage of the algorithm implies that, $w^tP^{\mu}_{f^{t-1}}w^0$. Thus, $w^tP^{\mu}_{f^{t-1}}w^0P^{\mu}_{f^{t-1}}w^t$ and we have a contradiction. Therefore $w\neq w^0$. Hence, Remark \ref{remark algoritmo} (ii)  implies that,  $\nu^t(w)R^{\mu}_w\nu^{t-1}(w)$, and the fact that $w\notin W^t$ implies that $\nu^t(w)=\nu^{t-1}(w)P^{\mu}_w{f^{t-1}}$. Therefore,  $\nu^t(w)P^{\mu}_w{f^{t-1}}$, contradicting \eqref{ecu -2 Lemma par bloqueante de vt-2}.
     
     \noindent \textbf{Case 2: $\boldsymbol{f\neq f^{t-1}.}$} By definition of $\nu^t$, $\nu^t(f)=\nu^{t-1}(f)$ together with the fact that $(w,f)$ is a blocking pair for  $\nu^{t}$ imply that
    \begin{equation}\label{ecu 0 Lemma par bloqueante de vt-1}
     wP^{\mu}_{f}\nu^{t}(f)=\nu^{t-1}(f)
    \end{equation}
    and 
    \begin{equation}\label{ecu 0 Lemma par bloqueante de vt-2}
      fP^{\mu}_{w}\nu^{t}(w).
    \end{equation}
    
If $\nu^{t-1}(f)=\mu(f)$, then by
\eqref{ecu 0 Lemma par bloqueante de vt-1}, $wP^{\mu}_{f}\mu(f)$ and this contradicts Remark \ref{remark de reduccion M-M} (i).  
Thus, there is a stage $t'<t-1$ of the algorithm such that $f^{t'}=f$, $\nu^{t'-1}(f)=\emptyset$ and $\nu^{t'}(f)=w^{t'}=\nu^{t-1}(f).$   This together with \eqref{ecu 0 Lemma par bloqueante de vt-1} imply that $wP^{\mu}_{f}w^{t'}$. Hence $w\notin W^{t'}$. Assume that $w=w^0$. Since $(w^0,f)$ blocks $\nu^{t}$ under $P^{\mu}$, $ w^0P^{\mu}_{f}\nu^{t}(f)=w^{t'}$. The fact that $t'$ is not the last stage of the algorithm implies that, $w^{t'}P^{\mu}_fw^0$. Thus, $w^{t'}P^{\mu}_fw^0P^{\mu}_fw^{t'}$ and we have a contradiction. Therefore, $w\neq w^0$. This implies that $\nu^{t'}(w)P^{\mu}_wf$ and, by Remark \ref{remark algoritmo} (ii), we have that $\nu^{t}(w)R^{\mu}_w\nu^{t'}(w)P^{\mu}_wf$. This last fact contradicts \eqref{ecu 0 Lemma par bloqueante de vt-2}.

Therefore,  either $(w,f)=(w^0,f^0)$ or $f=f^{t}.$
\end{proof}

\noindent \begin{proof}[Proof of Theorem \ref{nu es estable}]
\begin{enumerate}[(i)]
    \item  We proceed by induction.
First, we prove that  $W^1 \neq \emptyset.$ Since $w^0= \mu(f^0)\neq \mu_W(f^0)$, by the Single Agent Theorem\footnote{The Single Agent Theorem states that if an agent is single in a stable matching, then it is single in all stable matchings \citep{mcvitie1970stable}.}, there is $\overline{w}\in W\setminus \{w^0\}$ such that $\overline{w}=\mu_W(f^0).$ By the optimality of $\mu_W$ and the definition of $\nu^0$, $f^0=\mu_W(\overline{w})P_{\overline{w}}\mu(\overline{w})=\nu^0(\overline{w}).$ Thus, $f^0 P_{\overline{w}}\nu^0(\overline{w})$. By the preference reduction procedure, $f^0 P^{\mu}_{\overline{w}}\nu^0(\overline{w})$. Therefore $\overline{w}\in W^1$. Hence, $W^1\neq \emptyset.$ 

Next, assume that $W^{\widehat{t}} \neq \emptyset$ for each $\widehat{t}<t$. 
Thus, $W^{t-1}\neq \emptyset$. Since $w^{t-1}\neq w^0$ (since $t-1$ is not the last stage of the algorithm) we have that $w^{t-1}=\nu^{t-2}(f^{t-1})$. This implies that $\mu(f^{t-1})\neq \emptyset.$ Let  $w^{\star}\in W$ be such that $w^{\star}=\mu(f^{t-1}).$ By the working of the algorithm, $\nu^{t-1}(f^{t-1})=\emptyset.$ Thus, there is a stage of the algorithm in which ${w^{\star}}$ is no longer matched to $f^{t-1}.$ Let $j\in \{1,\ldots, t-1\}$ be the lowest index such that $\nu^j(f^{t-1})\neq \mu(f^{t-1}).$ Then,  $w^{\star} = w^j$. Hence, $f^{j-1} P^{\mu}_{w^j} \nu^{j-1}(w^j)=f^{t-1}.$ By the reduction procedure, we have that $\mu_W(w^{j})R^{\mu}_{w^j}f^{j-1}P^{\mu}_{w^j}f^{t-1}.$ Then,  $\mu_W(w^j)\neq \mu(w^j).$ Thus, as $\mu(w^j)=f^{t-1}$, $\mu_W(f^{t-1})\neq \mu(f^{t-1}).$ 

Let $\overline{w}\in W$ be such that $\overline{w}=\mu_W(f^{t-1}).$
First, assume that $\overline{w}=w^0.$ Therefore,
$$f^{t-1}=\mu_W(w^0)P^{\mu}_{w^0}\mu(w^0)=f^0$$ and the algorithm includes $w^0$ in $W^t$.  Second, assume that $\overline{w}\neq w^0.$ By the reduction procedure  we have that 
$f^{t-1}=\mu_W(\overline{w})R^{\mu}_{\overline{w}}\nu^{t-1}(\overline{w}).$ Since $\nu^{t-1}(f^{t-1})=\emptyset,$ then $f^{t-1}P^{\mu}_{\overline{w}}\nu^{t-1}(\overline{w})$ and $\overline{w}\in W^t.$
We conclude that $W^t\neq \emptyset.$

\item Assume that the algorithm does not stop in a finite number of stages.  Then, by the finiteness of the model, there are two stages of the algorithm $t$ and $t'$ with $t<t'$, a worker $w\in W$ such that $w=w^t=w^{t'}\neq w^0$, and a firm $f\in F$ such that $f=f^{t-1}=f^{t'-1}$ that satisfy 
\begin{equation}\label{ecu lema parate 1}
  \nu^{t-1}(f)=\emptyset \text{ and }\nu^{t}(f)=w,
\end{equation}
together with 
\begin{equation}\label{ecu lema parate 2}
   \nu^{t'-1}(f)=\emptyset \text{ and }\nu^{t'}(f)=w.
\end{equation}
 By \eqref{ecu lema parate 1} and \eqref{ecu lema parate 2}, we have that $\nu^{t}(w)=f$, and $\nu^{t'-1}(w)\neq f$, so $t<t'-1$. Then, by Remark \ref{remark algoritmo} (ii), $\nu^{t'-1}(w)P^{\mu}_w \nu^{t}(w)=f.$ Hence, $$f=\nu^{t'}(w)P^{\mu}_w \nu^{t'-1}(w) P^{\mu}_w \nu^t(w)=f.$$ Thus, $fP^{\mu}_wf,$ a contradiction. 
 Therefore, the algorithm stops in a finite number of stages.
 
\item 
Let $\nu$ be the matching obtained by the algorithm.  Let $T$ be the stage in which the algorithm stops and let $\nu^0,\ldots,\nu^{T-1},\nu$ and $f^0,\ldots,f^T$ the sequences of matchings and firms involved in the pairs defined by the algorithm, respectively.  By Remark \ref{remark algoritmo} (i), the pair $(w^0,f^0)$ is blocking pair for $\nu^{t}$ under $P^{\mu}$ with $t=1,\ldots,T-1$. By Lemma \ref{par bloquenate de nu^t tambien de nu^t-1}, the firms involved in  blocking pairs for $\nu^{t}$ are $f^{0}$ (only with worker $w^0$) and $f^{t}$, for $t=1,\ldots,T-1$.   By definition of $\nu$, it is only left to prove that neither $f^{0}$ nor $f^{T-1}$ can be members of a blocking pair under $P^{\mu}$. First, by definition of $\nu$ and $w^{0}\in W^T$, $\nu(w_0)P^\mu_{w^0} \mu(w^{0})=f^{0}$ implying that $(w^{0},f^0)$ is not a blocking pair for $\nu$ under $P^{\mu}$. Lastly, since $w^{0}\in W^T$ we have that $\nu(f^{T-1})=w^{0}P^{\mu}_{f^{T-1}}w$ for each $w\in W^T.$
 Therefore, $\nu$ is a stable matching under $P^{\mu}$.
 \end{enumerate}
 \end{proof}

\begin{lemma}\label{w0f0 en ultimo ciclo}
   Let $\sigma^k$ be the last cycle of any sequence of cycles generated from $\mu$ to $\nu$. Let $T$ be the stage of the algorithm such that $(w^0,f^{T-1})\in A^T(P^\mu,w^0)
   $. Then, $(w^0,f^{T-1})\in \sigma^{k}.$
\end{lemma}
\begin{proof}
   Let $\sigma^k$ be the last cycle of any sequence of cycles generated from $\mu$ to $\nu$. Let $T$ be the stage of the algorithm such that $(w^0,f^{T-1})\in A^T(P^\mu,w^0)
   $.\footnote{Note that, by the stopping criterion of the algorithm, $T$ is the last stage.} Assume that  $(w^0,f^{T-1})\notin \sigma^{k}.$ Let $t$ be a stage of the algorithm such that $(f,w)\in\sigma^k$, $f$ is the last firm of $\sigma^k_F$ that is unmatched at stage $t-1$ and matched to $w$  at stage $t$. This means  that $\nu^{t-1}(f)=\emptyset$, $\nu^{t}(f)=w=\nu(f)=\mu^k(f)$ and for each $f'\in\sigma^k_F\setminus ¨{f}$, $\nu^{t-1}(f')=\nu(f')=\mu^k(f')$. Assume that $w=w^0 $, then $w^0=\mu^k(f)=\nu(f).$ Hence, by the output of the algorithm, $f=f^{T-1}$. Thus $(w^0,f^{T-1})\in\sigma^k$ and this is a contradiction. Therefore, $w\neq w^0$ and there is $\overline{f}\in F\setminus\sigma^k_F$ such that $\nu^{t-1}(\overline{f})=w$, and $\nu^{t}(\overline{f})=\emptyset.$ Hence,  $wP^{\mu}_{\overline{f}}\nu(\overline{f}).$ Since $\mu^{k}=\nu$, and  $\overline{f}\notin \sigma^k_F$, then 
    \begin{equation}\label{ecu 1 prop w_0 en ultimo ciclo}
        wP^{\mu}_{\overline{f}}\nu(\overline{f})=\mu^{k}(\overline{f})=\mu^{k-1}(\overline{f}).
    \end{equation}
Since $f= \nu^{t}(w)$ and $\overline{f}=\nu^{t-1}(w)$, by Remark \ref{remark algoritmo} (ii), we have that $fP^{\mu}_w\overline{f}.$ Since $w\in \sigma^{k}_W$, then there is
\begin{equation}\label{ecu 2 prop w_0 en ultimo ciclo}
   \widetilde{f}=\mu^{k-1}(w).
\end{equation}
Let $\widetilde{w}=\mu^{k}(\widetilde{f}).$ Thus,  $w=\mu^{k-1}(\widetilde{f})$, and $\widetilde{w}=\mu^{k}(\widetilde{f})=\nu(\widetilde{f})$. Hence, 
\begin{equation}\label{ecu 3 prop w_0 en ultimo ciclo}
    w=\mu^{k-1}(\widetilde{f})P^{\mu}_{\widetilde{f}}\mu^{k}(\widetilde{f})=\widetilde{w}.
\end{equation}
Then, there is a stage of the algorithm, say $\widetilde{t}$, such that  
$\nu^{\widetilde{t}-1}(\widetilde{f})=\emptyset$, and $\nu^{\widetilde{t}}(\widetilde{f})=\widetilde{w}$, i.e $(\widetilde{w},\widetilde{f}) \in A^{\widetilde{t}}(P^\mu,w^0).$
By \eqref{ecu 3 prop w_0 en ultimo ciclo}, and the fact that $\nu^{\widetilde{t}}(\widetilde{f})=\widetilde{w},$ we have that $w\notin W^{\widetilde{t}}$. Then, $\nu^{\widetilde{t}-1}(w)P^{\mu}_w\widetilde{f}.$ By definition of $t$, and the fact that $\widetilde{f}\in\sigma^{k}_F$, we have that $\widetilde{t}\leq t-1.$ This implies, by Remark \ref{remark algoritmo} (ii), that
\begin{equation}\label{ecu 4 prop w_0 en ultimo ciclo}
    \overline{f}= \nu^{t-1}(w)R^{\mu}_w \nu^{\widetilde{t}}(w)=\nu^{\widetilde{t}-1}(w)P^{\mu}_w \widetilde{f}.
\end{equation}
By \eqref{ecu 2 prop w_0 en ultimo ciclo} and \eqref{ecu 4 prop w_0 en ultimo ciclo}, we have that
\begin{equation}\label{ecu 5 prop w_0 en ultimo ciclo}
    \overline{f}P^{\mu}_w\widetilde{f}=\mu^{k-1}(w).
\end{equation}
By \eqref{ecu 1 prop w_0 en ultimo ciclo} and \eqref{ecu 5 prop w_0 en ultimo ciclo}, we have that the pair $({w,\overline{f}})$ blocks stable matching $\mu^{k-1}$, a contradiction. Therefore,  $(w^0,f^{T-1})\in \sigma^{k}.$
\end{proof}

\bigskip

\noindent \begin{proof}[Proof of Lemma \ref{el ultimo ciclo de dos suceciones es el mismo}]
   Assume that there are two different sequences of cycles generated from $\mu$ to $\nu.$ Let $\sigma^k$ and $\sigma^{k'} $ be the last cycle of each sequence. Let $T$ be the last stage of the algorithm. By Lemma \ref{w0f0 en ultimo ciclo}, we have that $(w^0,f^{T-1})\in \sigma^k \cap \sigma^{k'}$. Consider the reduced preference profiles $P^{\mu^{k-1}}$ and $P^{\mu^{k'-1}}$ from which  cycles  $\sigma^k$ and $\sigma^{k'} $ are computed. Since  $(w^0,f^{T-1})\in \sigma^k \cap \sigma^{k'}$, we have that, by definition of cycles, the second worker of both $P^{\mu^{k-1}}_{f}$ and $P_f^{\mu^{k'-1}}$ is $\mu^{k-1}_{\sigma^k}(f)=\mu^{k'-1}_{\sigma^{k'}}(f)=\nu(f)$ for each $f\in \sigma^k\cup \sigma^{k'}.$ Therefore,  $\sigma^k= \sigma^{k'}.$
\end{proof}

\bigskip

\noindent\begin{proof}[Proof of Lemma \ref{cada par del algoritmo en un ciclo}]
    Let $(w^{t},f^{t-1})\in A^t(P^{\mu},w^0).$ Since all pairs formed by the algorithm are acceptable under $P^\mu$, and given that we assume that all acceptable pairs under $P^\mu$ are stable under $P^\mu$, by Theorem 4 in \cite{gusfield1987three},\footnote{In a model in which all acceptable pairs are stable, this theorem says that if $(w,f)$ is a stable pair such that $w\neq \mu_W(f)$, then $(w,f)$ belongs to exactly one cycle. This theorem was originally presented in \cite{irving1986complexity}.} $(w^{t},f^{t-1})$ belongs to exactly one cycle. Assume that this cycle is not in the sequence. Then, $\nu(f^{t-1})P^\mu_{f^{t-1}}w^t$. Let $w^\star=\nu(f^{t-1})$. By definition of $W^{t}$, $w^\star \notin W^t$. Therefore, $\nu^{t-1}(w^\star)P^\mu_{w^\star}f^{t-1}=\nu(w^\star)$. This contradicts Remark \ref{remark algoritmo} (ii), implying that the cycle belongs to the sequence.   
\end{proof}

\begin{lemma}\label{cada ciclo es tocado por el algoritmo}
    Let $\{\sigma^1,\ldots,\sigma^k\}$ be a sequence of cycles generated from $\mu$ to $\nu$ . Then for each $\sigma^j$, there is a stage of the algorithm, say $t$, such that $(w^t,f^{t-1}) \in A^t(P^\mu,w^0)$ and $(w^t,f^{t-1})\in\sigma^j$.
    
\end{lemma}
\begin{proof}
 Note that, by Lemma \ref{w0f0 en ultimo ciclo}, for $\sigma^k$ (the last cycle of the sequence generated from $\mu$ to $\nu$) the result holds.  Now, fix a cycle of a sequence generated from $\mu$ to $\nu$, say $\sigma^j$ ($j\neq k$). Assume that $\sigma^j \cap A(P^\mu,w^0)=\emptyset$.  Let $t'$ be the first stage of the algorithm such that $\nu^{t'}(f)=\emptyset$  for some $(w,f)\in \sigma^j
   .$  This implies that there are $\widetilde{w},w^{t'+1}\in W$ such that
   $$
       \nu^{t'-1}(f)=   \widetilde{w},
   \text{ and }
           \nu^{t'+1}(f)=w^{t'+1}.
 $$

Recall that by definition of cycle $\sigma^{j}$, $\mu^j(f)=w.$  By the assumption that $\nu^t(f)\neq \mu^{j}(f)$ for each stage $t$ of the algorithm, we have that  $w^{t'+1}\neq w.$ Since $t'$ is the first stage of the algorithm such that $\nu^{t'}(f)=\emptyset$, then $wP^{\mu}_{f}w^{t'+1}.$ Hence, $w\notin W^{t'+1}.$ This implies that $\nu^{t'}(w)P^{\mu}_{w}f.$ Then, $\nu^{t'}(w)\neq \mu^{j-1}(w).$ Let $f'$ be such that $f'\in \sigma_j^F$, and $f'=\mu^{j-1}(w).$ Then, there is a stage $\tilde{t}<t'$ such that   $\nu^{\tilde{t}-1}(f')=w$ and  $\nu^{\tilde{t}}(f')=\emptyset.$ This contradicts our election of $t'.$  Therefore,  $\sigma^j \cap A(P^\mu,w^0)\neq\emptyset$.
   \end{proof}
   
\begin{remark}
    Note that both Lemmata \ref{w0f0 en ultimo ciclo} and \ref{cada ciclo es tocado por el algoritmo} hold without the assumption that all acceptable pairs are stable.
\end{remark}
\bigskip

\noindent \begin{proof}[Proof of Lemma \ref{cardinalidad del primer ciclo y ultimo ciclo}]Let $\{\sigma^1,\ldots,\sigma^k\}$ be any sequence of cycles generated from $\mu$ to $\nu$.
    
     To prove (i), by Remark \ref{remark algoritmo} (iii) it is sufficient to show that $\mu(w)\neq \nu(w)$ for each $w\in \sigma^k_W.$  Note that, by Theorem \ref{asignaciones consec}, there is a sequence of stable matchings $\mu^{0},\ldots,\mu^k$ such that $\mu^0=\mu$, $\sigma^{i}$ is a cycle for $P^{\mu^{i-1}}$ and $\mu^k=\nu.$   Then, $\nu(w)=\mu^{k}(w)P^{\mu}_w\mu^{k-1}(w)$ for each $w\in \sigma^k_W$. Moreover, by the construction of the sequence of matchings, we have that $\mu^{k-1}R^{\mu}_W \mu$. Thus,  $\mu(w)\neq \nu(w)$ for each $w\in \sigma^k_W.$  Therefore, $(w,\nu(w))\in A(P^{\mu},w^0)$ for each $(w,\nu(w))\in \sigma^{k}.$

 To prove (ii),  assume $k\neq 1$  and ,w.l.o.g.,  that $\sigma^1=\left\{(w_2,f_1),(w_3,f_2),\ldots, (w_1,f_r)\right\}$. By Lemma \ref{cada ciclo es tocado por el algoritmo}, we have that $A(P^\mu,w^0)\cap \sigma^1\neq \emptyset$. Let $t$ be the first stage of the algorithm in which a pair of $\sigma^1$ is formed. W.l.o.g., assume  $(w_2,f_1)$ is such a pair, i.e. $\nu^{t}(f_1)=w_2$. Moreover, $\nu^t(f_2)=\emptyset$. Since $(w_3,f_2)\in \sigma^1,$ and $\mu(w_3)=f_3$, then $f_2 P^\mu_{w_3}f_3.$ Notice that $\mu(w_3)=\nu^{t}(w_3)$ imply that $f_2 P^\mu_{w_3}\nu^{t}(w_3)=f_3$. Thus, $w_3\in W^{t+1}.$  Since, $w_3$ is the second worker in $P^\mu_{f_2}$ (by definition of cycle), then $w_3$ is the most preferred worker in  $W^{t+1}$, and therefore, $\nu^{t+1}(f_2)=w_3.$ This means that the pair $(w_3,f_2)\in A^{t+1}(P^{\mu},w^0)\setminus A^{t}(P^{\mu},w^0)$. Since $\nu^{t+1}(f_3)=\emptyset$, repeating this reasoning, we can prove that the pair $(w_4,f_3)\in A^{t+2}(P^{\mu},w^0)\setminus A^{t+1}(P^{\mu},w^0)$. In each of the subsequent stages, the following pairs in the cycle are formed until pair $$(w_r,f_{r-1})\in A^{t+|\sigma^1|-2}(P^{\mu},w^0)\setminus A^{t+|\sigma^1|-3}(P^{\mu},w^0).$$
However, $(w_1,f_r)\notin A(P^\mu,w^0).$ To see this, let $f^{t-2}=\nu^{t-1}(w_1)$. By the election of the pair $(w_2,f_1)$,  we have $(w_1,f^{t-2})\notin \sigma^{1}.$ Since $(w_1,f_r)\in \sigma^1$, by definition of cycle,  $f^{t-2}P^\mu_{w_1}f_r.$ This implies that, by Remark \ref{remark algoritmo} (ii),  $(w_1,f_{r})\notin A(P^\mu,w^0).$  Therefore, $|A(P^\mu,w^0)\cap \sigma^1|=|\sigma^1|-1.$   
\end{proof}

\bigskip

Given that, by Corollary \ref{corolario nu estable}, $\mu P_F \nu$, we assure that $P^{\mu,\nu}$ is well defined.
\begin{lemma}\label{A (P) iguales}
$A(P^{\mu},w^0)=A(P^{\mu,\nu},w^0)$.  
\end{lemma}
\begin{proof}
   By Remark \ref{remark algoritmo} (ii), for each worker, the firms that belong to a pair formed by the algorithm are found in the worker's preference between his partner under $\nu$ and his partner under $\mu$. This means that  $\nu(w)R^{\mu}_w\nu^t(w)R^{\mu}_w\mu(w)$ for each $w\in W\setminus\{w^0\}$. Similarly, for each firm, the workers that belong to a pair formed by the algorithm are found in the firm's preference between its partner under $\mu$ and its partner under $\nu$. 
     Moreover, since each acceptable pair is a stable pair we have that $ w P^{\mu}_f w' R^{\mu}_f \nu(f)$ if and only if $w P_f^{\mu,\nu}w'$ for each $f \in F$ and each $w,w' \in W$. Symmetrically,  $\nu(w)R^{\mu}_w f P^{\mu}_w f'$ if and only if $ f P_w^{\mu,\nu}f'$ for each $w\in W$ and each $f,f'\in F$.
 Therefore, the algorithm forms the same pairs in both preferences $P^{\mu}$ and $P^{\mu,\nu}$. 
\end{proof}

\bigskip

\begin{lemma}\label{lemma de AP de un reducido subconjunto del AP general}
    Let $\sigma^{i}$ be a cycle of a sequence of cycles generated from $\mu$ to $\nu$. Then, $A(P^{\mu^{i},\nu},w^0)\subseteq A(P^{\mu,\nu},w^0).$
\end{lemma}
\begin{proof}
   Assume that there is a pair  $(w,f)\in A(P^{\mu^{i},\nu},w^0) $ and  $(w,f)\notin A(P^{\mu,\nu},w^0) $. Let $t'$ the stage of the algorithm applied to the profile $ P^{\mu^{i},\nu}$ and worker $w^0$ such that $A^{t}(P^{\mu^{i},\nu},w^0)= A^{t}(P^{\mu,\nu},w^0)$ for each $t<t'$ and $(w,f) \in A^{t'}(P^{\mu^{i},\nu},w^0).$  We denote by $\nu^{t}  $ and $\widetilde{\nu}^{t}  $ to the matchings formed at the stage $t$ in the algorithm applied  to the profile $P^{\mu,\nu}$ and $P^{\mu^{i},\nu}$, respectively. Note that, for each $t<t'$ we have that $\widetilde{\nu}^{t}=\nu^{t}$. Moreover, in stage $t'-1$ we have that $\widetilde{\nu}^{t'-1}(f)=\nu^{t'-1}(f)=\emptyset$, and in stage $t'$, $\widetilde{\nu}^{t'}(f)=w$ and $\nu^{t'}(f) \neq w.$
   
Note that, since $f=\widetilde{\nu}^{t'}(w)$, by Remark \ref{remark algoritmo} (ii) we have that 
\begin{equation}\label{nutilde en t menos uno preferido a f}
    fP^{\mu^{i},\nu}_{w}\widetilde{\nu}^{t'-1}(w).
\end{equation}
Let $=w^\star=\nu^{t'}(f) .$
Now, there are two cases to consider:
   
   \noindent \textbf{Case 1: $\boldsymbol{wP^{\mu,\nu}_f w^\star}$}. Then, $w\notin W^{t'}$ (when the algorithm is applied to $P^{\mu,\nu}$). Let $f'=\nu^{t'-1}(w).$ Thus,  $f'P^{\mu,\nu}_w f.$ Since $f=\widetilde{\nu}^{t'}(w)$ and $f'=\widetilde{\nu}^{t'-1}(w)$, we have that $f$ and $f'$ are acceptable firms for $w$ under $P_w^{\mu_{i},\nu}$ . Therefore, by the reduction procedure, $f'P_w^{\mu_{i},\nu}f$. Thus, $\widetilde{\nu}^{t'-1}(w)P^{\mu^{i},\nu}_w f,$
   contradicting $\eqref{nutilde en t menos uno preferido a f}.$
 
  \noindent \textbf{Case 2: $\boldsymbol{w^\star P^{\mu,\nu}_f w}$}. The proof is similar to that of Case 1.

\noindent Since in both cases we get a contradiction, we conclude that $A(P^{\mu^{i},\nu},w^0)\subseteq A(P^{\mu,\nu},w^0).$
\end{proof}

\bigskip

\noindent \begin{proof}[Proof of Theorem \ref{teorema principal}] Let $ \mu\in S(P)$, $w^0\in W$ such that $\mu(w^0)\neq \mu_{W}(w^0)$, and let $\nu$ be the output of the algorithm applied to the profile $P^{\mu}$ and worker $w^0$. Let $\{\sigma^1,\ldots,\sigma^k\}$ the sequence of cycles generated from $\mu$ to $\nu$, and  let  $\mu=\mu^0,\mu^1,\ldots,\mu^{k}=\nu$ be the sequence of stable matchings between $\mu$ and $\nu$ generated by the sequence of cycles. Recall that, by Theorem \ref{asignaciones consec} and Lemma \ref{el ultimo ciclo de dos suceciones es el mismo},  any sequence of cycles generated from $\mu$ to $\nu$ has the same cycles (possibly in different order) and the last one is always the same, so the right-hand side of equation \eqref{nombre} is well defined.  If $k=1$,  by Lemma \ref{cardinalidad del primer ciclo y ultimo ciclo} (i) the proof is completed.
     If $k=2$,  by Lemma \ref{cardinalidad del primer ciclo y ultimo ciclo} (i) and (ii) the proof is completed. Now, we prove the case $k\geq3$.
Note that, by Lemma \ref{cardinalidad del primer ciclo y ultimo ciclo} (i), $|A(P^{\mu^{k-1},\nu},w^0)\cap\sigma^{k}|=|\sigma^k|.$ Since $\sigma^{k-1}$ is a cycle for $P^{\mu^{k-2},\nu}$, by Lemmata \ref{cada ciclo es tocado por el algoritmo} and  \ref{lemma de AP de un reducido subconjunto del AP general}, $A(P^{\mu^{k-1},\nu},w^0)\subsetneq A(P^{\mu^{k-2},\nu},w^0)$. Thus, by Lemma \ref{cada par del algoritmo en un ciclo}, $ A(P^{\mu^{k-2},\nu},w^0)\setminus A(P^{\mu^{k-1},\nu},w^0)\subseteq \sigma^{k-1}.$ Notice that cycle $\sigma^{k-1}$ is the first cycle of the sequence generated from $\mu^{k-2}$ to $\nu$. Then, by Lemma \ref{cardinalidad del primer ciclo y ultimo ciclo} (ii), $|A(P^{\mu^{k-2},\nu},w^0)|=|\sigma^{k}|+(|\sigma^{k-1}|-1).$ 

Following the same reasoning, by Lemmata \ref{cada ciclo es tocado por el algoritmo} and \ref{lemma de AP de un reducido subconjunto del AP general}, $A(P^{\mu^{k-2},\nu},w^0)\subsetneq A(P^{\mu^{k-3},\nu},w^0)$. Thus, by Lemma \ref{cada par del algoritmo en un ciclo}, $ A(P^{\mu^{k-3},\nu},w^0)\setminus A(P^{\mu^{k-2},\nu},w^0)\subseteq \sigma^{k-2}.$ As $\sigma^{k-2}$ is the first cycle of the sequence generated from $\mu^{k-3}$ to $\nu$, then by Lemma \ref{cardinalidad del primer ciclo y ultimo ciclo} (ii), $|A(P^{\mu^{k-3},\nu},w^0)|=|\sigma^{k}|+(|\sigma^{k-1}|-1) +(|\sigma^{k-2}|-1).$ Repeating this argument we achieve that  $|A(P^{\mu^{0},\nu},w^0)|=|\sigma^k|+\sum_{i=1}^{k-1}(|\sigma^i|-1).$ Therefore, by Lemma \ref{A (P) iguales}, 
 equation \eqref{nombre} holds.
\end{proof}
\section{Further results}\label{apendiceb}

In this appendix, we present two examples. One of them illustrates that when we allow for some acceptable pairs that are not stable, we only find a lower bound for the number of steps needed to re-stabilize the market. The other example shows that the extension to a many-to-one model is not straightforward. Although our (adapted) algorithm can be applied and reaches a stable matching, our result on counting the steps for
re-stabilization using cycles in preferences is not true.

\subsection{When all acceptable pairs are not stable}\label{apendice sub b1}

\begin{example}
    Let $(F,W,P)$ be a market where $F=\{f_1,f_2,f_3,f_4\}$, $W=\{w_1,w_2,w_3,w_4\}$, and the preferences profile given by:
\begin{center}
    
$
\begin{array}{l}
P_{f_{1}}:w_{1},w_{2},w_3,w_4 \\ 
P_{f_{2}} :w_{2},w_{1} \\ 
P_{f_{3}} :w_{3},w_2 \\ 
P_{f_{4}} :w_4,w_1,w_3
\end{array}
\begin{array}{l}
P_{w_{1}} :f_{2},f_4,f_{1} \\ 
P_{w_{2}} :f_{3},f_{1},f_2 \\ 
P_{w_{3}} :f_{4},f_1,f_3 \\ 
P_{w_{4}} :f_1,f_4
\end{array}\hspace{25pt}$
\end{center}

\noindent Noto that preference profile $P$ is already reduced, i.e. $P=P^{\mu_F}.$ The stable matchings are

\begin{center}
$
\begin{array}{l|cccc}
&f_1&f_2&f_3&f_4  \\ \hline
\mu_F& w_1&w_2&w_3&w_4\\
\mu_1& w_2&w_1&w_3&w_4\\
\mu_2& w_3&w_1&w_2&w_4\\
\mu_W& w_4&w_1&w_2&w_3\\
\end{array}$
\end{center}

\noindent Note that the $(w_1,f_4)$ is an acceptable pair but not stable. Now, we apply the algorithm to the initial matching $\mu_F$, the profile $P^{\mu_F}$, and consider $w^0=w_4.$ Define $f^0=\mu_F(w_4)=f_4$, $A^0=\emptyset,$ and 
$$
\nu^0=\begin{pmatrix}
f_1 & f_2 & f_3 & f_4& \emptyset  \\
w_1& w_2 & w_3& \emptyset & w_4\\
 \end{pmatrix}.$$
\medskip
\noindent Now, we detail each stage of the algorithm:
\newline
\noindent \textit{\textbf{Stage 1}}
As $W^1=\left\lbrace w_1,w_3\right\rbrace $ and $w_1P^{\mu_F}_{f_4}w_3$, then $ w^1=w_1$, and $ f^1=\nu^0(w_1)=f_1.$ Since $ w^1=w_1\neq w^0=w_4,$ we have 
$A^1=\left\lbrace (w_1,f_4)\right\rbrace $, and
$$
\nu^1=\begin{pmatrix}
f_1 & f_2 & f_3 & f_4& \emptyset  \\
\emptyset& w_2 & w_3& w_1 & w_4\\
\end{pmatrix}.$$ 
\bigskip
\noindent \textit{\textbf{Stage 2}}
As $W^2=\left\lbrace w_2,w_3,w_4\right\rbrace $ and $w_2P^{\mu_F}_{f_1}w_3P^{\mu_F}_{f_1}w_4$, then $ w^2=w_2$, and $ f^2=\nu^1(w_2)=f_2.$ Since $ w^2=w_2\neq w^0=w_4,$ we have $A^2=\left\lbrace (w_2,f_1),(w_1,f_4)\right\rbrace $, and 
$$
\nu^2=\begin{pmatrix}
f_1 & f_2 & f_3 & f_4& \emptyset  \\
w_2& \emptyset& w_3 & w_1 & w_4\\
\end{pmatrix}.$$ 
\bigskip
\noindent \textit{\textbf{Stage 3}}
As $W^3=\left\lbrace w_1\right\rbrace $, then $ w^3=w_1$, and $ f^3=\nu^2(w_1)=f_4.$ Since $ w^3=w_1\neq w^0=w_4$, we have $A^3=\left\lbrace (w_1,f_2),(w_2,f_1),(w_1,f_4)\right\rbrace $, and 
$$
\nu^3=\begin{pmatrix}
f_1 & f_2 & f_3 & f_4& \emptyset  \\
w_2& w_1& w_3 & \emptyset& w_4\\
\end{pmatrix}.$$ 

\noindent \textit{\textbf{Stage 4}}
As $W^4=\left\lbrace w_3\right\rbrace $, then $ w^4=w_3$, and $ f^4=\nu^3(w_3)=f_3.$ Since $ w^4=w_3\neq w^0=w_4$, we have $A^4=\left\lbrace (w_3,f_4),(w_1,f_2),(w_2,f_1),(w_1,f_4)\right\rbrace $, and 
$$
\nu^4=\begin{pmatrix}
f_1 & f_2 & f_3 & f_4& \emptyset  \\
w_2& w_1&  \emptyset& w_3& w_4\\
\end{pmatrix}.$$ 

\noindent \textit{\textbf{Stage 5}}
As $W^5=\left\lbrace w_2\right\rbrace $, then $ w^5=w_2$, and $ f^5=\nu^4(w_2)=f_1.$ Since $ w^5=w_2\neq w^0=w_4$, we have $A^5=\left\lbrace (w_2,f_3),(w_3,f_4),(w_1,f_2),(w_2,f_1),(w_1,f_4)\right\rbrace $, and 
$$
\nu^5=\begin{pmatrix}
f_1 & f_2 & f_3 & f_4& \emptyset  \\
\emptyset& w_1& w_2 & w_3& w_4\\
\end{pmatrix}.$$ 

\noindent \textit{\textbf{Stage 6}}
Lastly, as $W^6=\left\lbrace w_4\right\rbrace $, then $ w^6=w_4=w^0$. Therefore, the algorithm stops, and the outputs of the algorithm are  $A=\left\lbrace (w_4,f_1),(w_2,f_3),(w_3,f_4),(w_1,f_2),(w_2,f_1),(w_1,f_4)\right\rbrace $, and 
$$
\nu=\begin{pmatrix}
f_1 & f_2 & f_3 & f_4 \\
w_4& w_1& w_2 & w_3\\
\end{pmatrix}.$$

\noindent After six stages $\nu=\mu_W$. The unique succession of cycles from $\mu_F$ to $\mu_W$ is  $\left\lbrace \sigma^1,\sigma^2,\sigma^3\right\rbrace $ where 
$\sigma^1=\left\lbrace (w_2,f_1),(w_1,f_2)\right\rbrace$, $\sigma^2=\left\lbrace (w_3,f_1),(w_2,f_3)\right\rbrace$ and $\sigma^3=\left\lbrace (w_4,f_1),(w_3,f_4)\right\rbrace$.
Since $(w_1,f_4)$ is not a stable pair, Lemma \ref{cada par del algoritmo en un ciclo} does not hold. In this case, we have that $(w_1,f_4)\in A(P^{\mu_F},w^0)$ however  $(w_1,f_4)\notin \sigma^i$ for $1\leq i\leq 3.$ Therefore, $|A(P^{\mu_F},w^0)|=6>4=(|\sigma^1|-1)+(|\sigma^2|-1)+|\sigma^3|.$ Moreover, observe that $\sigma^1\subseteq A(P^{\mu_F},w^0)$ implying that Lemma \ref{cardinalidad del primer ciclo y ultimo ciclo} (ii) does not hold also. \hfill $\Diamond$
\end{example}

Note that in this example, $\sigma^1\subseteq A(P^{\mu_F},w^0)$. In fact, when not all acceptable pairs are stable, part (ii) of Lemma \ref{cardinalidad del primer ciclo y ultimo ciclo} should be change to: If  $k\neq1,$ then $|A(P^\mu,w^0)\cap \sigma^1|\geq|\sigma^1|-1.$
The proof is similar, but in this case, the pair $(w_1,f_r)$ could belong to $A(P^\mu,w^0)$. To see this, consider the case that $(w_1,f^{t-2})$ is  an acceptable but not stable pair. If $f^{t-2}P^\mu_{w_1}f_r$, following the same reasoning of the proof of Lemma \ref{cardinalidad del primer ciclo y ultimo ciclo} we get that   $|A(P^\mu,w^0)\cap \sigma^1|=|\sigma^1|-1.$ Now, consider the case that  $f_rP^\mu_{w_1}f^{t-2}P_{w_1}f_1=\mu(w_1)$. Recall that $(w_1,f^{t-2})\notin \sigma^1$. Since $(w_r,f_{r-1})\in A^{t+|\sigma^1|-2}(P^{\mu},w^0)\setminus A^{t+|\sigma^1|-3}(P^{\mu},w^0)$, we have that $\nu^{t+|\sigma^1|-2}(f_r)=\emptyset.$ Then, $\nu^{t+|\sigma^1|-1}(f_r)=w_1$ implying that $(w_1,f_r)\in A(P^\mu,w^0)$.

Therefore, when acceptable pairs are not necessarily stable, our Theorem \ref{teorema principal} can be reformulated into:
\bigskip

\emph{
 Let $ \mu\in S(P)$, $w^0\in W$ such that $\mu(w^0)\neq \mu_{W}(w^0)$, and let $\nu$ be the output of the algorithm applied to the profile $P^{\mu}$ and worker $w^0$. Then, 
    \begin{equation*}\label{ecu1teo}
        |A(P^\mu,w^0)|\geq|\sigma^k|+\sum_{i=1}^{k-1}(|\sigma^i|-1)
    \end{equation*}
    where $\{\sigma^1,\ldots,\sigma^k\}$ is any sequence of cycles generated from $\mu$ to $\nu$.}
\bigskip

The proof follows by an inductive reasoning similar to the one used in  the proof of Theorem \ref{teorema principal}.

\subsection{Extension to the many-to-one model}
In what follows, we introduce a modified version of our algorithm adapted to a many-to-one market, detailed in Table \ref{tabla algoritmo m-to-1}. Although the algorithm leads to a stable outcome, our result regarding counting the step for re-stabilization using cycles in preferences is not true.

\begin{table}[h!]
\centering
\begin{tabular}{l l}
\hline \hline
\multicolumn{2}{l}{\textbf{Algorithm:}}\vspace*{10 pt}\\

\textbf{Input} & A reduced profile $P^\mu$ and $w^0\in W$ such that $\mu(w^0)\neq\mu_W(w^0)$\\

\textbf{Output} & A set of worker-firm pairs $A$, and a matching $\nu\in S(P^\mu)$\vspace*{10 pt}\\
 & \hspace{-55pt}Define:\\
 & \hspace{-55pt} $f^0=\mu(w^0)$,\\
  & \hspace{-55pt} $ \nu^0  (f) =\left\{ 
\begin{array}{lcl}
 \mu(f)\setminus \{w^0\} &  & \text{if } f =f^0 \\ 
\mu (f)&  & \text{otherwise}  \\ 
\end{array}%
\right. ,
$\\
& \hspace{-55pt} $\nu^0(w)=\left\{ 
\begin{array}{lcl}
f  &  & \text{if } w\in \nu^0(f) \\ 
\emptyset &  & \text{otherwise}\\ 
\end{array}%
\right.,$ \\
& \hspace{-55pt} $A^0=\emptyset$.\\
\textbf{Stage $\boldsymbol{t\geq1}$} &  Let $W^t=\{w\in W\setminus\{w^0\}:f^{t-1}P^{\mu}_w \nu^{t-1}(w)\}$\\
& \hspace{20pt}\texttt{IF} $f^{t-1}P^{\mu}_{w^0} \mu(w^0)$\\
& \hspace{40 pt}\texttt{THEN} $W^t=W^t\cup\{w^0\}$\\
& Choose $w^t\in W^t$ such that $ w^t R^{\mu}_{f^{t-1}}w$ for each $w\in W^t$\\
& \hspace{20 pt}\texttt{IF} $w^t\neq w^0:$  \\
& \hspace{40 pt}\texttt{THEN} Define:\\
& \hspace{68 pt}$f^{t}=\nu^{t-1}(w^t)$  \\
&\hspace{68 pt}$
 \nu^t  (f) =\left\{ 
\begin{array}{lcl}
 \nu^{t-1}(f)\setminus\{w^t\} &  & \text{if } f =f^{t} \\ 
\nu^{t-1}(f) \cup \{w^t\}   &  & \text{if } f =f^{t-1} \\ 
\nu^{t-1} (f)&  & \text{otherwise}  \\ 
\end{array}%
\right. 
$\\
&\hspace{68 pt}$\nu^{t}(w)=\left\{ 
\begin{array}{lcl}
f  &  & \text{if } w\in \nu^t(f) \\ 
\emptyset &  & \text{otherwise}\\ 
\end{array}%
\right.$ \\
&\hspace{68 pt}$A^t=A^{t-1}\cup \{(w^t,f^{t-1})\}$\\
&\hspace{40 pt}\texttt{AND} continue to Stage $t+1.$\\
& \hspace{20 pt}\texttt{ELSE}: Set $
 \nu  (f) =\left\{ 
\begin{array}{lcl}
 \nu^{t-1}(f) \cup \{w^0\}   &  & \text{if } f =f^{t-1} \\ 
\nu^{t-1} (f)&  & \text{otherwise}  \\ 
\end{array}%
\right. 
$ \\
&\hspace{68 pt}$\nu  (w) =\left\{ 
\begin{array}{lcl}
f  &  & \text{if } w\in \nu(f) \\ 
\emptyset &  & \text{otherwise}\\ 
\end{array}%
\right. 
$\\
&\hspace{68 pt}$A=A^{t-1}\cup\{(w^0,f^{t-1})\}$, and \texttt{STOP}. \\


\hline \hline
\end{tabular}
\caption{Re-stabilization algorithm for many-to-one markets}
\label{tabla algoritmo m-to-1}
\end{table}

The following example illustrates how the modified algorithm re-stabilizes the market, but our result on counting steps fails. 

\begin{example}
    Let $(F,W,P)$ be a market where $F=\{f_1,f_2\}$, $W=\{w_1,w_2,w_3,w_4\}$, each firm has a quota $q_{f_i}=2$ for each $i=1,2$;  and the preference profile is given by:
\begin{center}
    
$
\begin{array}{l}
P_{f_{1}} :w_{1},w_{2},w_3,w_4 \\ 
P_{f_{2}} :w_3,w_4,w_{1},w_{2}
\end{array}
\begin{array}{l}
P_{w_{1}} =P_{w_{2}} :f_{2},f_{1} \\ 
P_{w_{3}} =P_{w_{4}} :f_{1},f_2
\end{array}\hspace{25pt}$
\end{center}

\noindent The stable matchings are
\begin{center}
$
\begin{array}{l|cccc}
&f_1&f_2 \\ \hline
\mu_F& \left\lbrace {w_1,w_2}\right\rbrace & \left\lbrace {w_3,w_4}\right\rbrace \\
\mu_1& \left\lbrace {w_2,w_3}\right\rbrace &\left\lbrace {w_1,w_4}\right\rbrace \\
\mu_W& \left\lbrace {w_3,w_4}\right\rbrace &\left\lbrace {w_1,w_2}\right\rbrace \\
\end{array}$
\end{center}

\noindent Observe that all acceptable pairs are also stable. Now, we apply the algorithm to the reduced profile $P^{\mu_F}$, and consider $w^0=w_2.$ Define $f^0=\mu_F(w_2)=f_1$, $A^0=\emptyset,$ and 
$$
\nu^0=\begin{pmatrix}
f_1 & f_2 & \emptyset  \\
\left\lbrace {w_1}\right\rbrace & \left\lbrace {w_3,w_4 }\right\rbrace &  \left\lbrace {w_2}\right\rbrace\\
 \end{pmatrix}.$$
\medskip
\noindent Now, we detail each stage of the algorithm:
\newline
\noindent \textit{\textbf{Stage 1}}
As $W^1=\left\lbrace w_3,w_4\right\rbrace $ and $w_3P^{\mu_F}_{f_1}w_4$, then $ w^1=w_3$, and $ f^1=\nu^0(w_3)=f_2.$ Since $ w^1=w_3\neq w^0=w_2,$ we have 
$A^1=\left\lbrace (w_3,f_1)\right\rbrace $, and
$$
\nu^1=\begin{pmatrix}
f_1 & f_2 & \emptyset  \\
\left\lbrace {w_1,w_3}\right\rbrace & \left\lbrace {w_4 }\right\rbrace &  \left\lbrace {w_2}\right\rbrace\\
\end{pmatrix}.$$ 
\bigskip
\noindent \textit{\textbf{Stage 2}}
As $W^2=\left\lbrace w_1,w_2\right\rbrace $ and $w_1P^{\mu_F}_{f_2}w_2$, then $ w^2=w_1$, and $ f^2=\nu^1(w_1)=f_1.$ Since $ w^2=w_1\neq w^0=w_2,$ we have $A^2=\left\lbrace (w_1,f_2),(w_3,f_1)\right\rbrace $, and 
$$
\nu^2=\begin{pmatrix}
f_1 & f_2 & \emptyset  \\
\left\lbrace {w_3}\right\rbrace & \left\lbrace {w_1,w_4 }\right\rbrace &  \left\lbrace {w_2}\right\rbrace\\
\end{pmatrix}.$$ 
\bigskip
\noindent \textit{\textbf{Stage 3}}
As $W^3=\left\lbrace w_4\right\rbrace $, then $ w^3=w_4$, and $ f^3=\nu^2(w_4)=f_2.$ Since $ w^3=w_4\neq w^0=w_2$, we have $A^3=\left\lbrace (w_4,f_1),(w_1,f_2),(w_3,f_1)\right\rbrace $, and 
$$
\nu^3=\begin{pmatrix}
f_1 & f_2 & \emptyset  \\
\left\lbrace {w_3,w_4}\right\rbrace & \left\lbrace {w_1 }\right\rbrace &  \left\lbrace {w_2}\right\rbrace\\
\end{pmatrix}.$$ 

\noindent \textit{\textbf{Stage 4}}
Lastly, as $W^4=\left\lbrace w_2\right\rbrace $, then $ w^4=w_2=w^0$. Therefore, the algorithm stops, and the outputs of the algorithm are $A=\left\lbrace (w_2,f_2)(w_4,f_1),(w_1,f_2),(w_3,f_1)\right\rbrace $, and 
$\nu=\begin{pmatrix}
f_1 & f_2  \\
\left\lbrace {w_3,w_4}\right\rbrace & \left\lbrace {w_1,w_2 }\right\rbrace\\
\end{pmatrix}.$

\noindent After four stages $\nu=\mu_W$. The unique sequence of cycles from $\mu_F$ to $\mu_W$ is  $\left\lbrace \sigma^1,\sigma^2\right\rbrace,$ where 
$\sigma^1=\left\lbrace (w_3,f_1),(w_1,f_2)\right\rbrace$ and $\sigma^2=\left\lbrace (w_4,f_1),(w_2,f_2)\right\rbrace$. Note that    $A(P^\mu,w_2)\cap \sigma^1 = \left\lbrace (w_1,f_2),(w_3,f_1)\right\rbrace =\sigma^1$, and $A(P^\mu,w_2)\cap \sigma^2 = \left\lbrace (w_2,f_2),(w_4,f_1)\right\rbrace =\sigma^2$ implying that $|A(P^\mu,w_2)|= |\sigma^1|+|\sigma^2|$. Thus, $|A(P^\mu,w_2)|> (|\sigma^1|-1)+|\sigma^2|$.\hfill $\Diamond$\end{example}

\end{document}